\documentclass[12pt]{amsart}
\usepackage{epsfig}
\usepackage{amssymb}
\usepackage{latexsym}
\usepackage{comment,caption,subcaption}
\usepackage{graphics} 

\newtheorem{theorem}{Theorem}[section]
\newtheorem{remark}[theorem]{Remark}
\newtheorem{lemma}[theorem]{Lemma}
\newtheorem{proposition}[theorem]{Proposition}

\newtheorem{corollary}[theorem]{Corollary}
\numberwithin{equation}{section}

\title[Persistence for Nested Infection Networks]{How Nested Infection Networks in Host-Phage Communities Come To Be}

\author[Daniel A. Korytowski and Hal L. Smith]{}

\subjclass{Primary: 92D25, 92D40.}
 \keywords{bacteriophage, competitive exclusion principle,  ecological succession, nested infection network, permanence, persistence, predator-mediated coexistence }

\email{Daniel.Korytowski@asu.edu}
 \email{halsmith@asu.edu}

\thanks{Both authors work was supported by NSF grant DMS-0918440}

\begin{document}

\maketitle

\centerline{\scshape Daniel A. Korytowski and Hal L. Smith}
\medskip
{\footnotesize
 \centerline{ School of Mathematical and Statistical Sciences}
 \centerline{Arizona State University}
   \centerline{Tempe, AZ 85287, USA}
}

\begin{abstract}
We show that a chemostat community of bacteria and bacteriophage in which bacteria compete for a single
nutrient and for which the bipartite infection network
is perfectly nested is permanent, a.k.a. uniformly persistent, provided that bacteria that are superior competitors
for nutrient devote the least to defence against infection and the virus that are the most efficient at infecting
host have the smallest host range. This confirms earlier work of Jover et al \cite{Jover} who raised the issue
of whether nested infection networks are permanent. In addition, we provide sufficient conditions that a bacteria-phage community of arbitrary size
with nested infection network can arise through a succession of permanent subcommunties each with a nested infection network by the
successive addition of one new population.
\end{abstract}

\section{introduction}

This work is inspired by the recent paper of Jover, Cortez, and Weitz \cite{Jover}.
Noting that empirical studies strongly suggest that the bipartite infection networks observed
in bacteria and virus communities tend to have a nested structure characterized by a hierarchy
among both host and virus strains which constrains which virus may infect which host, they identify key tradeoffs between
competitive ability of the bacteria hosts and defence against infection and, on the part of virus, between virulence and transmissibility
versus host range  such that
a nested infection network can be maintained.  They find that ``bacterial growth rate should decrease
with increasing defence against infection'' and that ``the efficiency of viral infection should decrease
with host range''. Their mathematical analysis of a Lotka-Volterra  model incorporating the above mentioned tradeoffs strongly suggests that the perfectly nested community structure
of $n$-host bacteria and $n$-virus is permanent, sometimes also called persistent, or uniformly
persistent \cite{HS,ST,T}. Indeed,
they establish several necessary conditions for permanence: (1) a positive equilibrium for the system with all host and virus populations at positive density exists,
and (2) every boundary equilibrium of the $2n$-dimensional ordinary
differential equations, where one or more population from the nested structure is missing, is unstable to invasion by at least one of the missing populations.
They also note that while equilibrium dynamics are rare for such systems, invasability of boundary equilibria can imply invasability
of general boundary dynamics provided permanence holds according to results of Hofbauer and Sigmund \cite{HS}. However,
permanence of a perfectly nested infection network is not established in \cite{Jover}. The famous example of
three-species competition described by May and Leonard \cite{ML} shows that the necessary conditions mentioned above are not sufficient for permanence.

Permanence of bacteriophage and bacteria in a chemostat has been established for mathematical models of very simple communities consisting of a single virus and one or two host bacteria in \cite{ST1,HS}.

A nested infection network of three bacterial strains and three virus strains has the structure described in the infection table below. An `x' in the matrix means that the host below is infected by the virus on the left while a blank entry indicates no infection; for example, the second column of three x's indicates that bacteria $H_1$ is infected by virus $V_1,V_2$ and $V_3$. Host $H_1$ is the least resistant to infection while $H_3$ is the most resistant; virus $V_1$ specializes
on a single host while $V_3$ is a generalist, infecting all host.

\vspace{0.5cm}

\begin{tabular}{|c|c|c|c|}
  \hline
  $V_3$ & x & x & x \\\hline
  $V_2$ & x & x &  \\\hline
  $V_1$ & x &  &  \\\hline
        & $H_1$ & $H_2$ & $H_3$ \\
  \hline
\end{tabular}

\vspace{0.5cm}
This community may have evolved by the sequential addition of one new population following a mutational event or the selection of a rare variant. Below, going back in time, we list in order the communities from which the one above may have evolved from an ancestral community consisting of a single bacteria and a single virus on the right.

\vspace{0.5cm}

\hbox{

\begin{tabular}{|c|c|c|c|}
  \hline
    $V_2$ & x & x &  \\\hline
  $V_1$ & x &  &  \\\hline
        & $H_1$ & $H_2$ & $H_3$ \\
  \hline
\end{tabular}

\hspace{10pt}

\begin{tabular}{|c|c|c|}
  \hline
    $V_2$ & x & x  \\\hline
  $V_1$ & x &   \\\hline
        & $H_1$ & $H_2$  \\
  \hline
\end{tabular}

\hspace{10pt}

\begin{tabular}{|c|c|c|}
  \hline
      $V_1$ & x &   \\\hline
        & $H_1$ & $H_2$  \\
  \hline
\end{tabular}

\hspace{10pt}

\begin{tabular}{|c|c|}
  \hline
      $V_1$ & x   \\\hline
        & $H_1$   \\
  \hline
\end{tabular}
}

\vspace{0.5cm}

Other possible evolutionary trajectories starting from the ancestral pair at the bottom are highly unlikely. Obviously, a new virus cannot evolve without there being a susceptible host for it; however, a new bacterial strain resistent, or partially resistent, to some virus may evolve. Obviously, the three-host, three-virus
network need not be the end of the evolutionary sequence. A fourth bacterial strain may evolve resistance to all three virus.

Just such a sequence of mutational or selection events is observed in chemostat experiments starting from a single bacteria population and a
single virus population and leading to a nested infection network. Chao et al \cite{CLS} describe such a scenario in their experimental
observations of {\it E. Coli} and phage $T7$. A bacterial mutant resistant to the virus is observed to evolve first.
Resistance is conferred by a mutation affecting a receptor on the host surface to which the virus binds. Subsequently, a viral mutant evolves which is able to infect both
bacterial populations. Eventually, another bacterial mutant arises which is resistant to both virus. Similar evolutionary scenarios are noted in the review
of Bohannan and Lenski \cite{BL}. Thus, a nested infection structure can evolve
as an arms race between host and parasite.

Our goal in this paper is to show that a nested infection network consisting of $n$ bacterial host and $n$ lytic virus is permanent
given the trade-offs identified in \cite{Jover}. Recall that permanence means that there is a positive threshold, independent of positive initial
conditions of all populations, which every bacteria and virus  density ultimately exceeds.

However, we replace the Lotka-Volterra model used by Jover et al \cite{Jover} by a chemostat-based model where bacterial populations compete for nutrient
and virus populations compete for hosts as in \cite{CLS,HS,ST1,WHL}, although we ignore latency of virus infection. Aside from the additional realism of including competition for nutrient, our model avoids
the non-generic bacterial dynamics of the Lotka-Volterra model which possesses an $n-1$-dimensional simplex of virus-free equilibria.

Chemostat-based models of microbial competition for a single nutrient are known to induce a ranking of competitive ability
among the microbes determined by their break-even nutrient concentrations for growth, here denoted by $\lambda$ but often by $R^*$ in the ecological literature.
The competitive exclusion principle applies: a single microbial population, the one with smallest $\lambda$, drives all others to extinction \cite{Ti,SW} in the absence of virus.
In our model of a nested infection network, this host can be infected by every virus strain and as the $\lambda$ value of
host strains increases (i.e., it becomes less competitive for nutrient)
it is subject to infection by fewer virus strains. Virus populations are ranked by their efficiency at infecting host.
The most efficient strain specializes on the host with smallest $\lambda$ and as infection efficiency decreases host range increases
so that the virus strain of rank $k$ infects the $k$ most competitive host strains.

Our permanence result is a dramatic example of predator-mediated coexistence. In the absence of phage, only a single bacterial strain
can survive. However, the addition of an equal number of phage to our microbial community with infection efficiency versus host range tradeoff as noted above lead to the
coexistence of all populations.

In fact, we will show that the $n$-bacteria, $n$-virus community can arise through a succession of permanent sub-communities just as described in the infection tables above for the case $n=3$, starting with an ancestral community of one susceptible bacterial host and one virus. This is important because it ensures
that the intermediate communities are sufficiently stable so as to persist until a fortuitous mutational or colonization event allows further progression.
Permanence is not a guarantee of long term persistence since environmental stochasticity may intervene to cause an extinction event, especially when a population
is in a low part of its cycle. See Figure~\ref{fig1} below. However, our permanence result implies that should an extinction event occur, the resulting community
is likely to be a permanent one and therefore recovery is possible.

We also show that time averages of species densities are asymptotic to appropriate equilibrium levels. Solutions of our chemostat-based model are highly
oscillatory, apparently aperiodic, just as those observed for the Lotka-Volterra system of Jover et al \cite{Jover}. See Figure~\ref{fig1}.

Perhaps it is interesting to note that the mathematical justification used to establish our results is to exploit the evolutionary sequence noted in the infection tables above by way of the principle of mathematical induction, establishing permanence in a given sub-community in the successional sequence by appealing to the permanence hypothesis of its predecessor in the sequence.

The competitive exclusion principle is critical to our approach.
We will show that two virus strains cannot share the same set of bacterial hosts (i.e. cannot have the same host range) since one of the virus
will be more efficient at exploiting the host and drive the other to extinction.  Similarly, two bacterial strains cannot suffer infection by the same set of virus
because the weaker competitor for nutrient will eventually be excluded. Therefore, the competitive exclusion principle drives the evolution of communities towards a nested infection structure.

As noted in \cite{Jover}, perfectly nested infection networks are  generally only observed for very small host-virus communities.
Because natural host-virus communities have strong tendency to be approximately nested in their infection structure,  it is worth while to
consider how the idealized nested network may have evolved. Mathematical modeling is especially useful for exploring these idealized scenarios.
Furthermore, permanence, or persistence in mathematical models is known to be robust to model perturbations  under appropriate conditions \cite{S,GH,HSZ} and therefore it should continue to hold
for small deviations from a nested infection structure.

\section{A Chemostat-based Host-Virus Model}

The standard chemostat model of microbial competition for a single limiting nutrient \cite{SW} is modified by adding lytic virus. Our model
is a special case of general host-virus models formulated in \cite{CLS} which include viral latency.
Let $R$ denote the nutrient which supports the growth of bacteria strains $H_i$; it is supplied at concentration $R_0$ from the feed. $V_i$ denote the various virus strains that parasitize
the bacteria. Bacteria strain $H_i$ is characterized by its specific growth rate $f_i(R)$ and its yield $\gamma_i$.
For simplicity, we assume that the yield is the same for all bacterial strains: $\gamma_i=\gamma$ is independent of $i$. At this point, we assume only that the specific growth rates  $f_i$ are increasing functions of nutrient $R$, vanishing when $R=0$.
Following \cite{Jover}, we assume that virus strain $V_i$ is characterized by its adsorption rate $\phi_i$ and its burst size $\beta_i$, both of which are assumed to be independent of which host strain it infects. $D$ denotes the dilution rate of the chemostat.

The community of bacterial strains $H_1, H_2, \cdots, H_n$ and virus strains $V_1, V_2, \cdots, V_n$ is structured as follows. Virus strain $V_i$ parasitizes all host strains $H_j$ for $j\le i$.
Thus, strain $V_1$ specializes on host $H_1$ while strain $V_n$ is a generalist, infecting all host strains. As $i$ increases, virus strain $V_i$
becomes more generalist, less of a specialist; the index $i$ is indicative of the number of host strains $V_i$ infects. This structure is referred to as a
nested infection network in \cite{Jover}.

Our model is described by the following differential equations:
\begin{eqnarray}\label{eqns}
  R' &=& D(R_0-R)-\sum_i \frac{1}{\gamma}f_i(R)H_i \nonumber \\
  H_i' &=& H_i(f_i(R)-D)-H_i\sum_{j\ge i}\phi_jV_j \\
  V_i' &=& \beta_i\phi_iV_i \sum_{j\le i}H_j-DV_i,\ 1\le i\le n. \nonumber
\end{eqnarray}

Non-dimensional quantities are identified below:
$$
N=R/R_0, \ B_i=H_i/(\gamma R_0),\ DP_i=\phi_iV_i,\ \tau=Dt
$$
Again using prime for derivative with respect to $\tau$, we have the equations
\begin{eqnarray}\label{scaledeqns}
  N' &=& 1-N-\sum_i g_i(N)B_i \nonumber\\
  B_i' &=& B_i(g_i(N)-1)-B_i\sum_{j\ge i}P_j \\
  P_i' &=& s_i^{-1} P_i \left(\sum_{j\le i}B_j-s_i\right),\ 1\le i\le n.\nonumber
\end{eqnarray}
where
$$
s_i=\frac{D}{\beta_i\phi_i\gamma R_0},\ g_i(N)=f_i(R_0 N)/D.
$$
Now, each virus strain is characterized by a single parameter $s_i$ which reflects its burst size $\beta_i$ and its
adsorption rate $\phi_i$. Clearly, smaller $s_i$ translates to stronger ability to exploit the host.

Following \cite{Jover}, we  assume that a virus with larger host range (generalist) has weaker ability to exploit its hosts than a specialist virus with small host range:
\begin{equation}\label{s}
s_1<s_2<s_3<\cdots < s_n
\end{equation}

Assume that the specific growth rate $g_i$ is a strictly increasing function of nutrient concentration and that there exists the break-even nutrient concentration $\lambda_i<1$ for strain $B_i$ defined by the balance of growth and dilution: $g_i(\lambda_i)=1$.
We assume that the bacterial species are ordered such that
\begin{equation}\label{lambda}
0<\lambda_1<\lambda_2<\cdots <\lambda_n<1.
\end{equation}
This implies that in the absence of virus, $B_i$ dominates $B_j$ if $i<j$ but that each bacteria is viable in the absence of the others.
Indeed, classical chemostat theory \cite{SW,Ti} implies that $B_1$ would eliminate all $B_j,\ j>1$ in the absence of the virus.
In particular, the superiority rank of a bacterial strain is inversely related to the number of virus strains that infect it.
Strain $B_1$ is the best competitor in virus-free competition for nutrient but it can be infected by all the virus strains, while strain $B_n$ is the worst
competitor for nutrient but can be infected only by virus strain $P_n$.

System \eqref{scaledeqns} enjoys the usual chemostat conservation principle, namely that the total nutrient content of bacteria and virus plus free nutrient
$$
T=N+\sum_i B_i+\sum_i s_iP_i
$$
must come into balance with the input of nutrient:
$$
T'=1-T.
$$
On the exponentially attracting invariant set $T=1$ we can drop the equation for $N$ from \eqref{scaledeqns} and replace $N$ by $1-\sum_i B_i-\sum_i s_iP_i$.

As a final model simplification, linear specific growth rates  $g_i(N)=r_i N$ are used  where, by \eqref{lambda}, we must have
\begin{equation}\label{r}
1<r_n<r_{n-1}<\cdots <r_2<r_1.
\end{equation}
Then $\lambda_i=1/r_i$. The result is the system with Lotka-Volterra structure
\begin{eqnarray}\label{LV}
    B_i' &=& r_iB_i\left(1-\frac{1}{r_i}-\sum_j (B_i+s_iP_i)\right)-B_i\sum_{j\ge i}P_j \\
  P_i' &=& s_i^{-1} P_i \left(\sum_{j\le i}B_j-s_i\right)\nonumber
\end{eqnarray}
$U=\sum_i (B_i+s_iP_i)$  represents the nutrient value of the bacteria and virus. It satisfies
\begin{equation}\label{U}
U'=W-(1+W)U,\quad W=\sum_i r_iB_i
\end{equation}
We consider the dynamics of \eqref{LV} on the positively invariant set
\begin{equation}\label{omega}
\Omega=\{(B_1,\cdots,B_n,S_1,\cdots, S_n)\in \mathbb{R}^{2n}_+:\sum_i (B_i+s_iP_i)\le 1\}
\end{equation}

\section{Equilibria}

It is well-known that in the absence of virus, there are only  single-population bacterial equilibria for chemostat systems. See \cite{SW}.
Let $E_i=(1-\lambda_i)e_i$ denote the equilibrium where host strain $B_i$ is alone. Here, $e_i$ is the unit vector with all components zero except the $i$th which is one.
In the absence of virus, $E_1$ attracts all solutions with $B_1(0)>0$.

Next we consider equilibria where all or nearly all host and virus are present.

\begin{proposition}\label{positivequil}
There exists an equilibrium $E^*$ with $B_i$ and $P_i$ positive for all $i$ if and only if
\begin{equation}\label{ss}
\frac{r_n}{1+Q_n}>1
\end{equation}
where $Q_1=r_1s_1$ and
$$Q_n=s_1(r_1-r_2)+s_2(r_2-r_3)+\cdots + s_{n-1}(r_{n-1}-r_n)+s_nr_n,\ n>1.$$

In fact,
\begin{eqnarray}\label{equil}
B_1^*&=& s_1,\ B_j^*=s_j-s_{j-1},\ j>1,\\
P_j^*&=&\frac{r_j-r_{j+1}}{1+Q_n},\ j<n,\ P_n^*=\frac{r_n}{1+Q_n}-1.\nonumber
\end{eqnarray}
The positive equilibrium $E^*$ is unique and $\sum_i B_i=s_n$.
Summing by parts yields $Q_n=\sum_{i=1}^n r_iB_i^*$.

\eqref{ss} also implies the existence of a unique equilibrium $E^\dag$ with all components positive except for $P_n=0$.
In fact,
\begin{eqnarray}\label{equiln}
B_j^\dag&=&B_j^*,\ 1\le j<n,\nonumber \\
B_n^\dag&=&B_n^*+(1-\frac{1+Q_n}{r_n}),\\
P_j^\dag&=&P_j^* \left(\frac{1+Q_n}{r_n}\right),\ j<n,\ P_n^\dag=0.\nonumber
\end{eqnarray}

\end{proposition}

\begin{remark}\label{remark}
\eqref{ss} is equivalent to
\begin{equation}\label{sss}
s_1(\frac{r_1-r_2}{r_n})+s_2(\frac{r_2-r_3}{r_n})+\cdots +s_{n-1}(\frac{r_{n-1}-r_n}{r_n})+s_n< 1-1/r_n,
\end{equation}
implying that $s_n<1$. To see that \eqref{r}, \eqref{s}, and \eqref{ss} can be satisfied simultaneously, note that
if the $r_i$ are chosen satisfying \eqref{r}, then one could choose $s_n$ such that $s_nr_1<r_n-1$. This implies that
\eqref{sss} holds with all $s_i=s_n$. In order to satisfy \eqref{s} it suffices to re-choose the $s_i,\ i<n$,
smaller so that \eqref{s} holds. Then \eqref{sss} will remain valid with the new $s_i$.
\end{remark}

\begin{remark}\label{Qn}
$Q_n=Q_{n-1} + r_nB_n^*$ which together with \eqref{r} implies that
$\frac{r_k}{1+Q_k}>\frac{r_n}{1+Q_n}$ for $1\le k<n$. Therefore, \eqref{ss} implies the existence of a
unique family of equilibria $E^*_k$ with $B_j,P_j=0, \ j>k$ described by \eqref{equil} but with
$Q_k$ replacing $Q_n$. Another family of equilibria, $E^\dag_k$, exists with $B_j=0, \ j>k$ and $P_j=0,\ j\ge k$  described by \eqref{equiln} but
with $Q_k$ replacing $Q_n$.
\end{remark}

\begin{remark}\label{equilibriumrelation}
Not surprisingly, the density of $B_i$ at the positive equilibrium $E^*$ is less then the density of $B_i$ at its equilibrium $E_i$.
More explicitly, $s_1<1-\frac{1}{r_1}$ and $s_i-s_{i-1}<1-\frac{1}{r_i},\ i>1$. This can be seen by rewriting \eqref{sss} as
$s_1r_1+r_2(s_2-s_1)+\cdots+r_n(s_n-s_{n-1})<r_n-1$
and using \eqref{r}. Note also that $P_j^\dag<P_j^*$, $B_j^\dag=B_j^*$ for $j<n$ and $B_n^\dag>B_n^*$.
\end{remark}

\begin{remark}\label{nutrient}
Free nutrient levels at $E^*$ and $E^\dag$ are revealing. At $E^\dag$, the (scaled) free nutrient level is given by $\lambda_n=1/r_n$, the same as
at $E_n$ where only bacteria strain $B_n$ is present with no virus. At $E^*$, the nutrient level is greater than at $E^\dag$. It is given by $\frac{1}{1+Q_n}$,
thus the ratio of the nutrient levels is precisely \eqref{ss}. Chao et al \cite{CLS} refer to $E^*$ as a ``phage-limited'' community while $E^\dag$ is referred to as
a ``nutrient-limited" one when $k=1$.
\end{remark}

\begin{remark}\label{unstable}
\eqref{ss} implies that $E^\dag$ is unstable to invasion by $P_n$ since
$$
\frac{s_nP_n'}{P_n}|_{E^\dag}=(1-\frac{1+Q_n}{r_n})>0.
$$
The additional nutrient level available at $E^\dag$ facilitates the invasion of the virus $P_n$.
\end{remark}

There are other equilibria. A complete list of them is given below. However, we will not have need of these details.

\begin{lemma}\label{equilibria}
Let $E$ be an equilibrium with at least one $P_i>0$. Then there exists some $k$ with $1\le k\le n$ such that
$E$ has exactly $k$ nonzero virus components and either  $k$ or $k+1$ nonzero bacteria components. Moreover, if we denote by
$I=\{i_1,i_2,\cdots, i_k\}$ the ordered indices with $P_i>0\Leftrightarrow i\in I$, then there exist a set $J=\{j_1,j_2,\cdots, j_k\}$
uniquely determined by $B_j>0,\ j\in J$ and by
\begin{equation}\label{order}
  j_1\le i_1<j_2\le i_2<j_3\le i_3 \cdots \le i_{k-1}<j_k\le i_k.
\end{equation}
If there are $k+1$ positive bacterial components, then $i_k<n$ and there exists $j_{k+1}>i_k$ such that $B_{j_{k+1}}>0$.

Moreover, if \eqref{ss} holds, for every such $k$ and any such ordered set $I=\{i_1,i_2,\cdots, i_k\}$ and any corresponding set $J=\{j_1,j_2,\cdots, j_k\}$
as in \eqref{order}, there exists a unique equilibrium $E$ where $P_i>0\Leftrightarrow i\in I$ and $B_j>0\Leftrightarrow j\in J$ having exactly $k$ nonzero
virus and $k$ nonzero bacteria.

The only equilibria without any virus present are the $E_i\equiv (1-1/r_i)e_i,\ 1\le i\le n$ with only $B_i>0$.
\end{lemma}

Figure~\ref{fig1} provides illuminating simulations of \eqref{LV} for the case $n=3$.
Parameter values are $r_1=3.2, r_2=3.1, r_3=3.0; s_1=0.1, s_2=0.15, s_3=0.2$.
For the top row, all initial population densities are given by
$B_i(0)=P_i(0)=0.1$. Observe that free nutrient level is high in this case because $P_3$, the dominant
virus, keeps  $B_3$ at low density. The bacterial community is ''phage limited'' in this case.
In the second row, initial data are $B_i(0)=0.1, 1\le i\le 3$, $P_j(0)=0.1, j=1,2$ and $P_3(0)=0$. Observe that
free nutrient levels are much lower than for the top row because $B_3$ is free to consume it.
The bacterial community is ''nutrient limited'' in this case.
In the third row has initial data are $B_i(0)=P_i(0)=0.1, 1\le i\le 2$  and $B_3(0)=P_3(0)=0$.

\begin{figure}
        \centering
        \begin{subfigure}[b]{0.3\textwidth}
                \includegraphics[width=\textwidth]{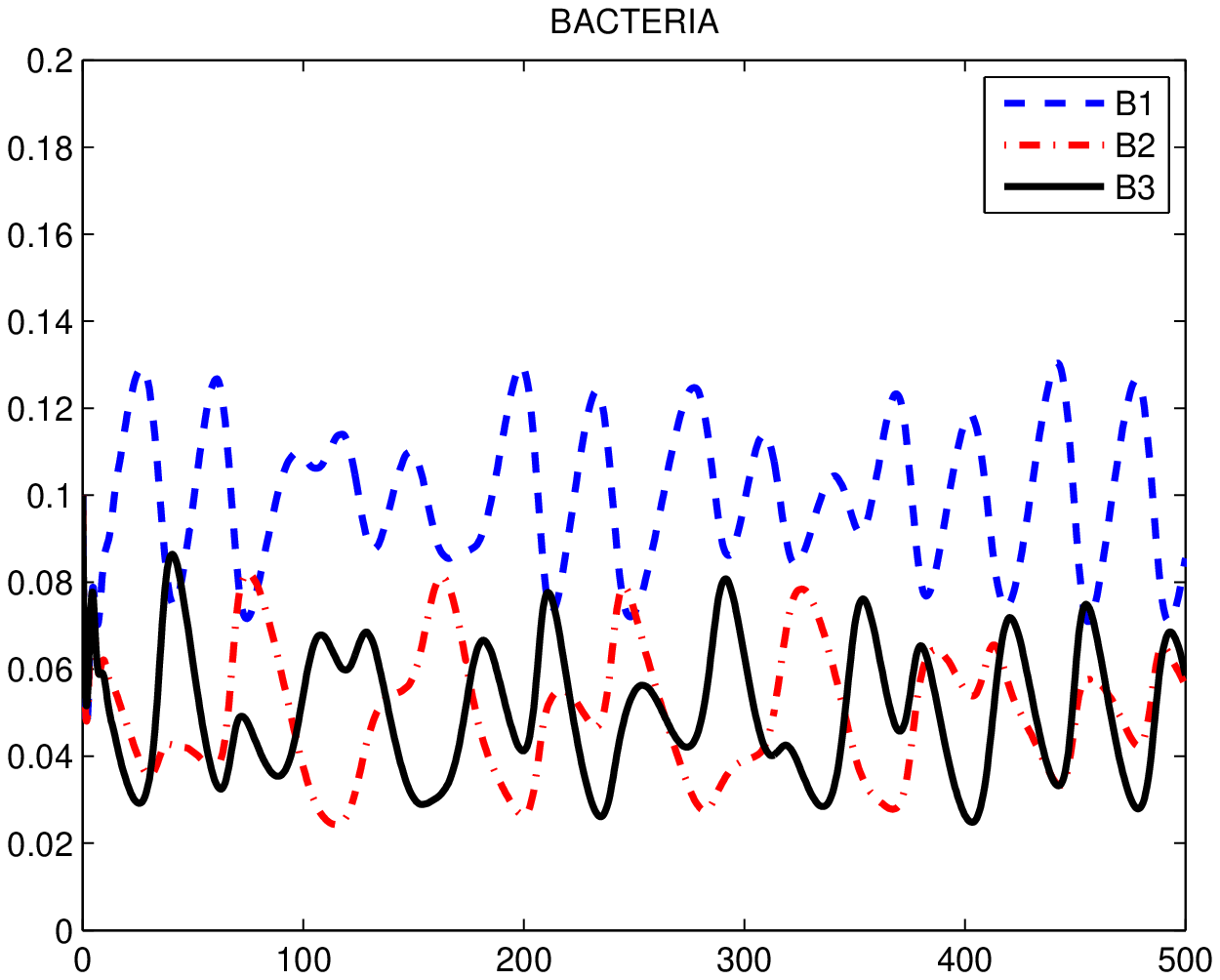}
        \end{subfigure}%
        ~ 
        \begin{subfigure}[b]{0.3\textwidth}
                \includegraphics[width=\textwidth]{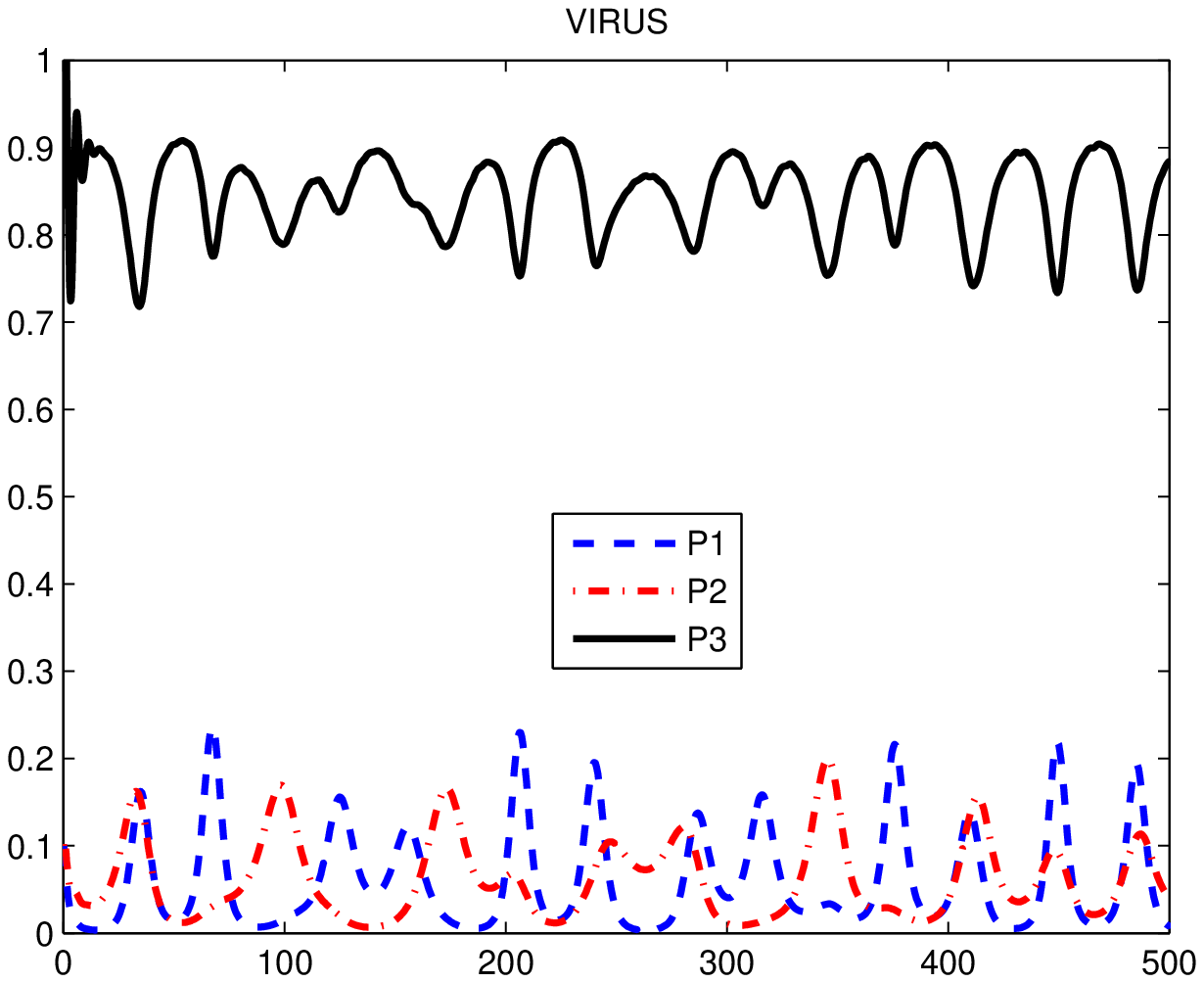}
        \end{subfigure}
        ~ 
        \begin{subfigure}[b]{0.3\textwidth}
                \includegraphics[width=\textwidth]{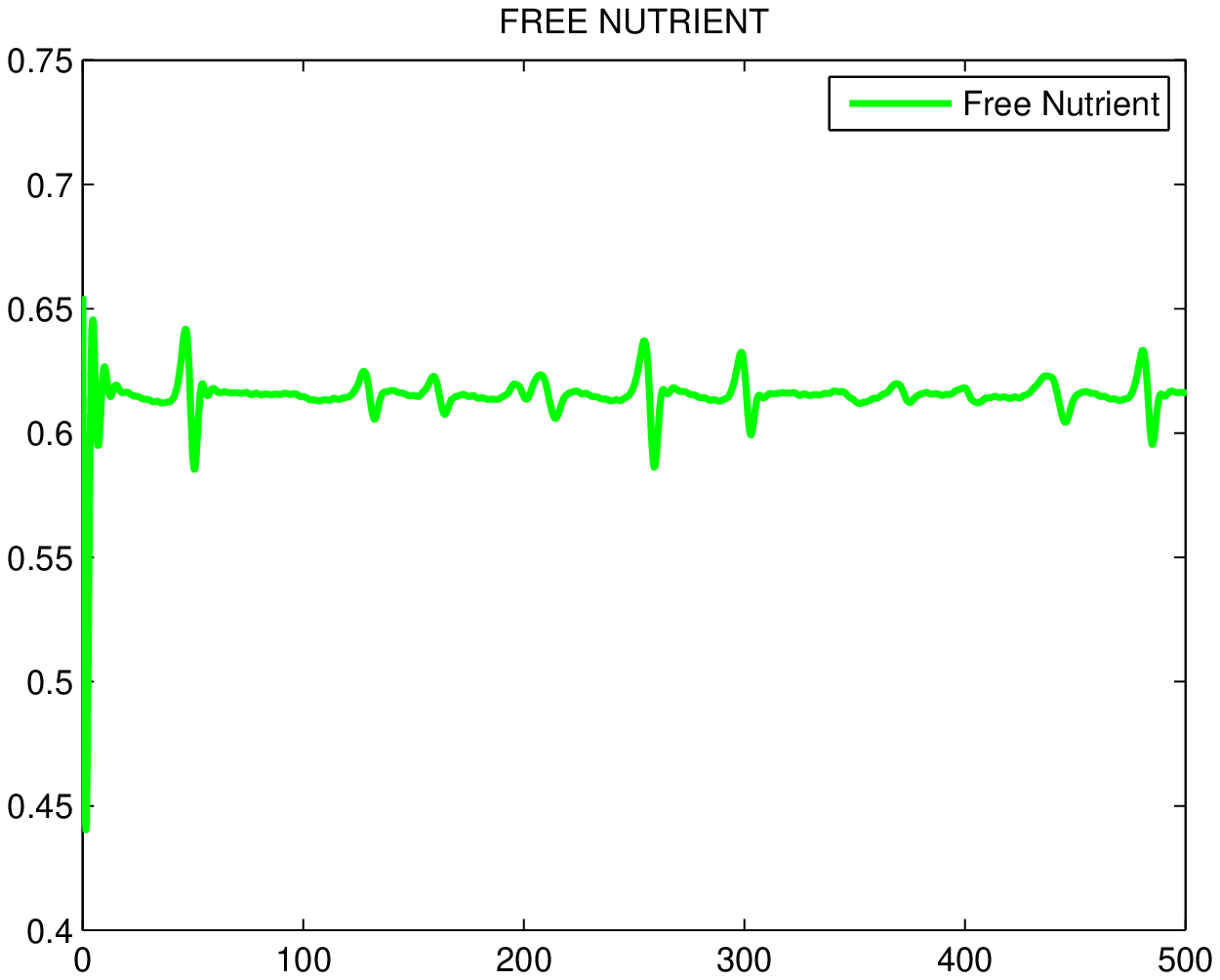}
        \end{subfigure}

        \begin{subfigure}[b]{0.3\textwidth}
                \includegraphics[width=\textwidth]{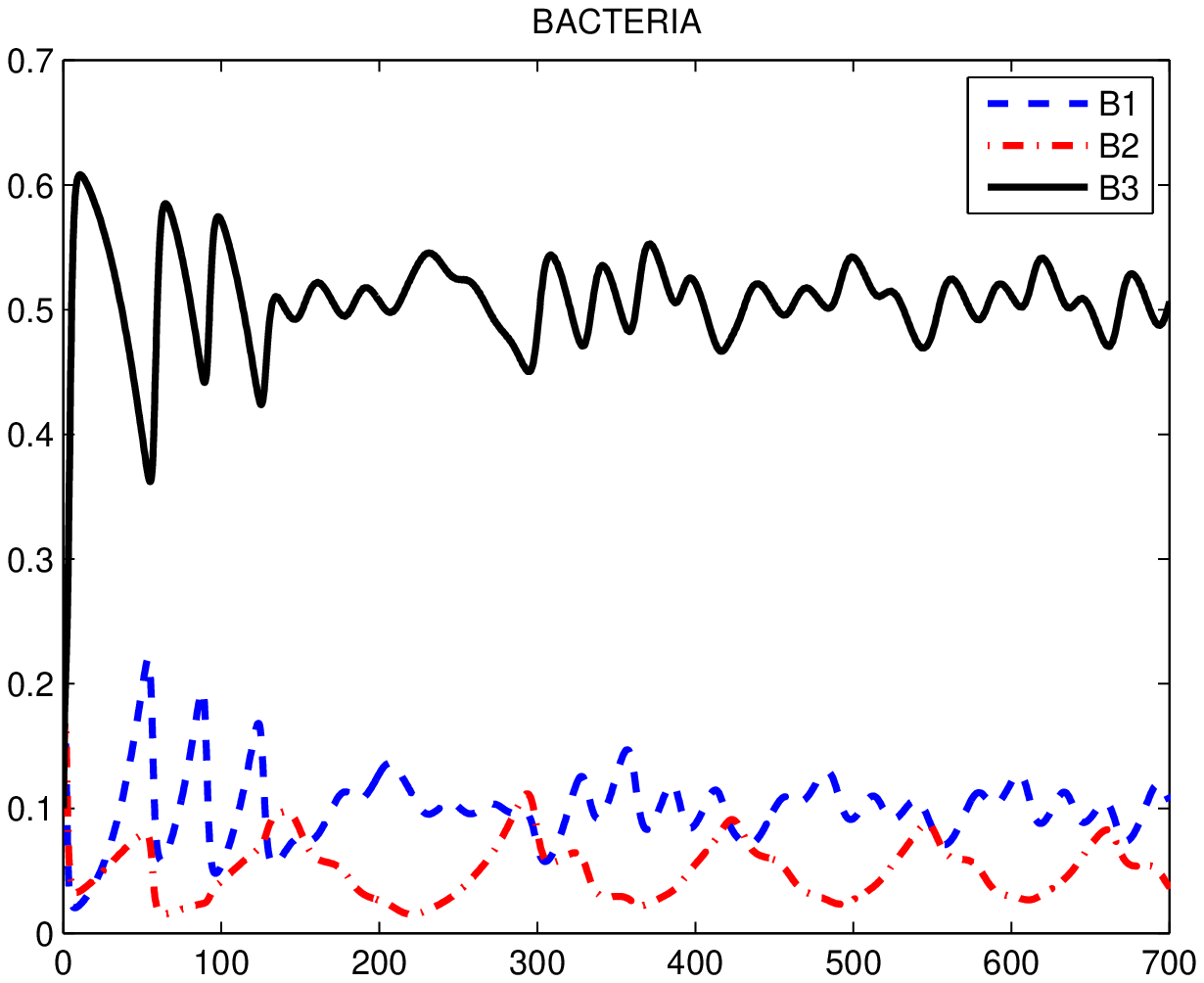}
        \end{subfigure}%
        ~ 
        \begin{subfigure}[b]{0.3\textwidth}
                \includegraphics[width=\textwidth]{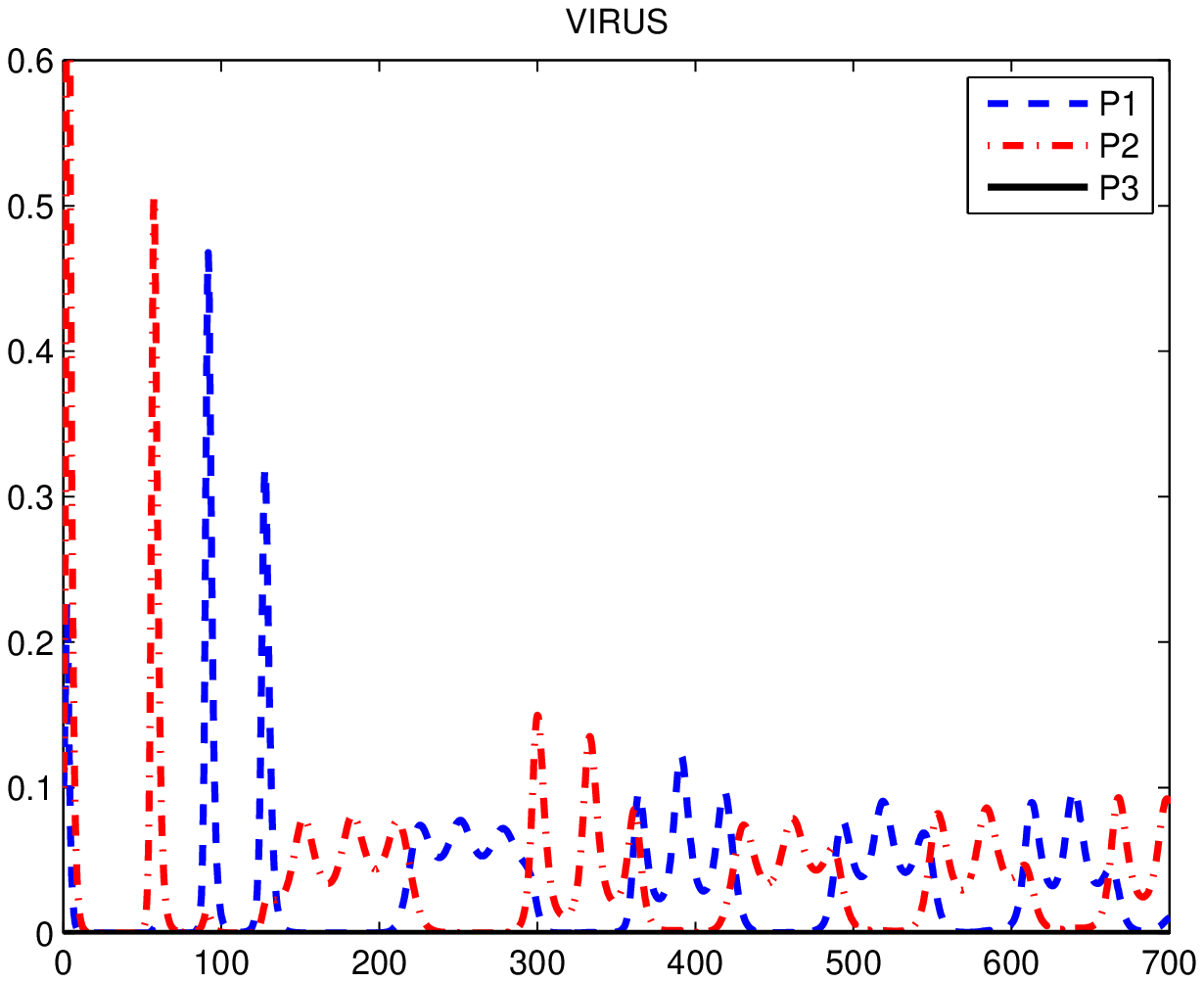}
        \end{subfigure}
        ~ 
        \begin{subfigure}[b]{0.3\textwidth}
                \includegraphics[width=\textwidth]{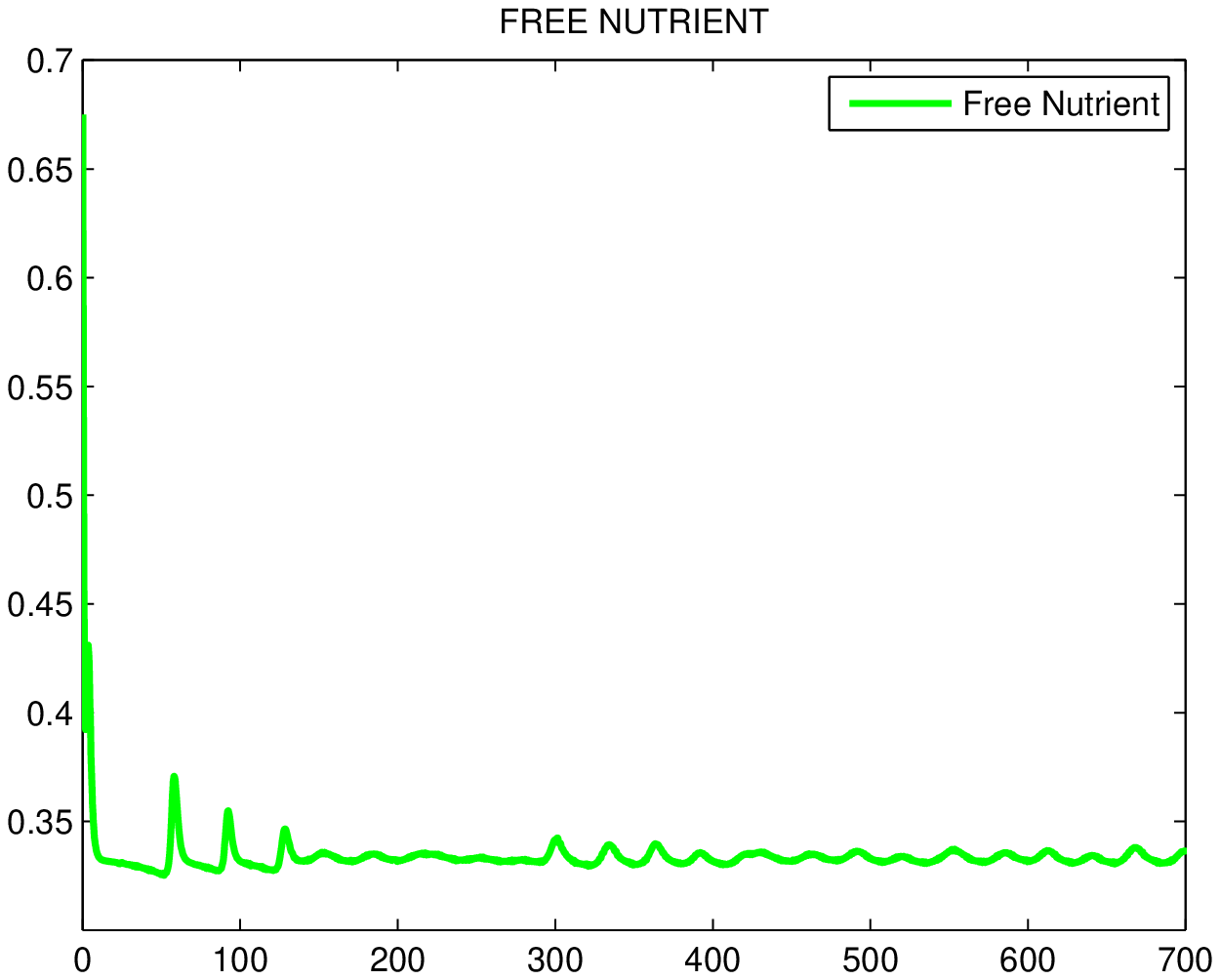}
        \end{subfigure}

        \begin{subfigure}[b]{0.3\textwidth}
                \includegraphics[width=\textwidth]{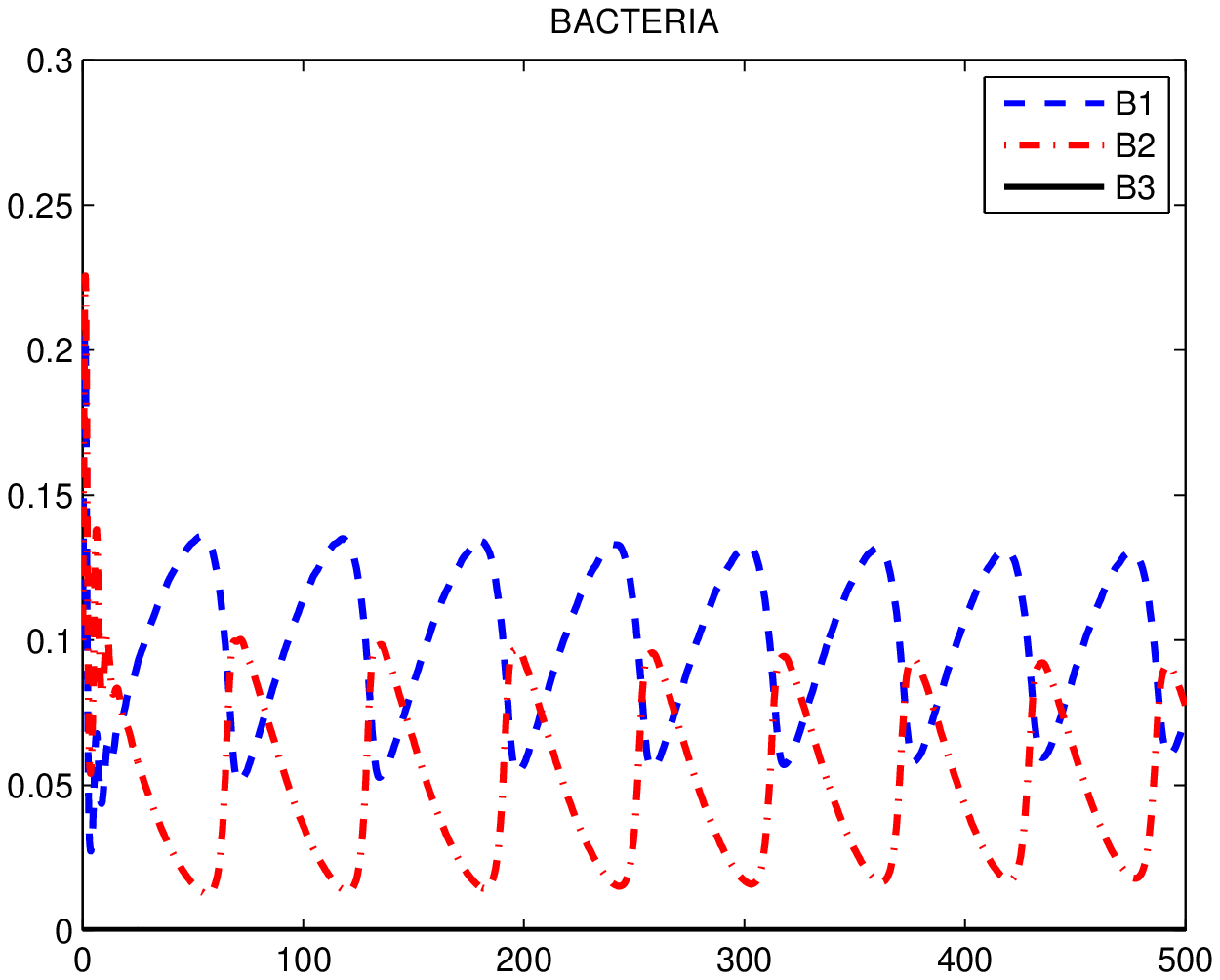}
        \end{subfigure}%
        ~ 
        \begin{subfigure}[b]{0.3\textwidth}
                \includegraphics[width=\textwidth]{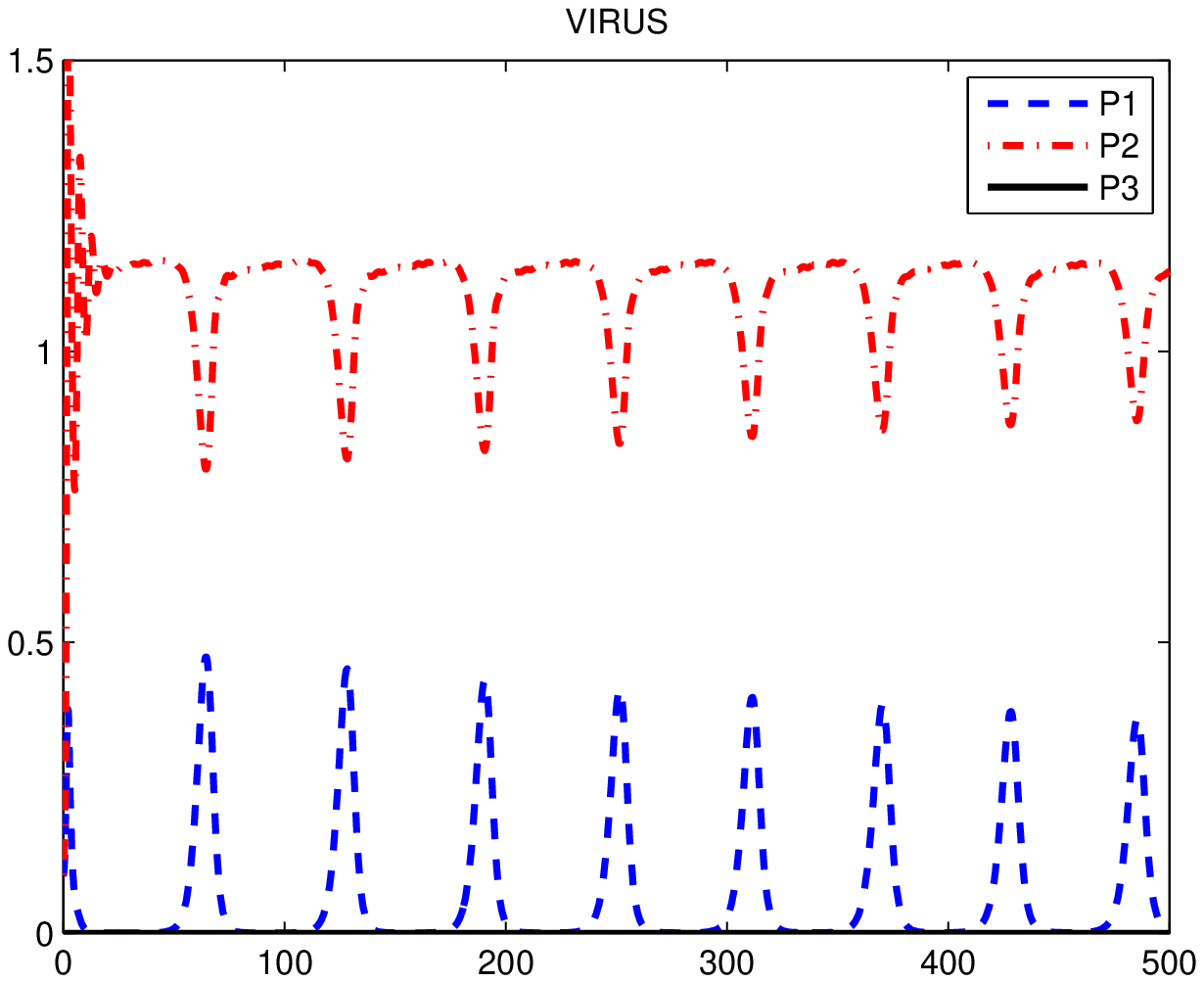}
        \end{subfigure}
        ~ 
        \begin{subfigure}[b]{0.3\textwidth}
                \includegraphics[width=\textwidth]{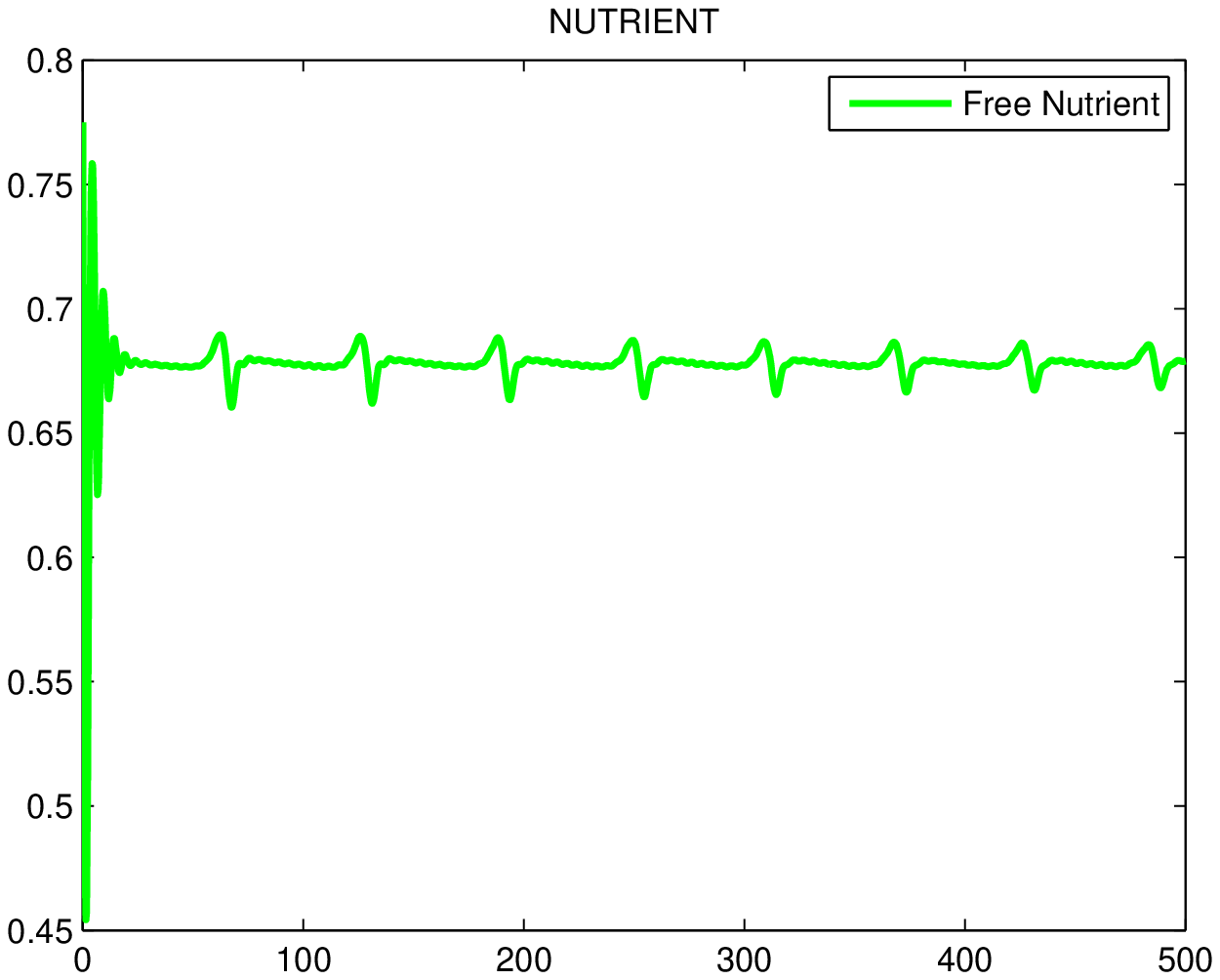}
        \end{subfigure}
        \vspace{2cm}
        \caption{Top row: $B_i(0)=P_i(0)=0.1, 1\le i\le 3$;
        second row: same except $P_3\equiv 0$; bottom row, as above except $B_3=P_3\equiv 0$. }\label{fig1}
\end{figure}

\section{Permanence}

In this section we state and prove our main results, Theorem~\ref{persist} and Corollary~\ref{meanvalue}. We begin by
establishing a competitive exclusion principle in the context of our model.

Two virus strains cannot share the same set of host bacterial strains; the weaker virus strain, the one with largest index, is doomed to extinction.
This is due to our assumption that each virus strain does not distinguish among the host that it infects in terms of adsorption rate or burst size.
Similarly, two bacteria strains
cannot share the same set of infecting virus strains; the bacterial strain which is the least competitive for nutrient, the one with largest index, is doomed to extinction. The next result formalizes these conclusions.

\begin{lemma}[Competitive Exclusion Principle]\label{extinction}
Let $1\le i<j\le n$.

If $P_i(0)>0, P_j(0)>0$ and $B_k(0)=0,\ i<k\le j$, then $P_j(t)\to 0$ as $t\to\infty$.

If $B_i(0)>0, B_j(0)>0$ and $P_k(0)=0,\ i\le k<j$, then $B_j(t)\to 0$ as $t\to\infty$.
\end{lemma}

\begin{proof}
The first assertion follows from Lemma 3.1 \cite{TS}, applied to the equations for $P_i$ and $P_j$, where $\kappa_1=1/s_i,\ \kappa_2=1/s_j,\ h_1(t)=h_2(t)\equiv -1$, and
$g(t)=\sum_{k\le i}B_k(t)=\sum_{k\le j}B_k(t)$. Note that
$$
s_j\int_0^t h_2(s)ds-s_i\int_0^t h_1(s)ds=(s_i-s_j)t\to-\infty.
$$
Hence the result follows from the quoted result since $P_i(t)$ is bounded.

The second assertion also follows from Lemma 3.1, applied to the equations for $B_i$ and $B_j$, where $\kappa_1=\kappa_2=1$,
$g(t)=\sum_{k\ge i}P_k(t)=\sum_{k\ge j}P_k(t)$, and $h_1=r_i-1-r_i U,\ h_2=r_j-1-r_j U$
where $U=\sum_k (B_k+s_kP_k)$. Note that
$$
\int_0^t h_2(s)-h_1(s) ds=(r_j-r_i)\int_0^t (1-U(s))ds\to -\infty,
$$
where the concluded limit is due to Lemma~\ref{Ubound} and \eqref{r}.
\end{proof}

\begin{remark}
Lemma~\ref{extinction} strongly constrains the evolution of bacteria and virus communities, at least under our assumption
that virus do not distinguish among their host in terms of adsorption rate and burst size. For example if a community consisting of a single
virus strain and a single bacterial strain is invaded by a new virus strain then either the resident virus strain or the invader must
be driven to extinction. However, our community can be successfully invaded by a new bacterial strain which is resistent to the virus
but an inferior competitor for nutrient than the resident.
\end{remark}

Hereafter, we assume without further mention that \eqref{ss} holds.

If $h:(a,\infty)\to \mathbb{R}$, we write $h_\infty=\liminf_{t\to\infty}h(t)$ and $h^\infty$ with limit superior in place
of limit inferior.

\begin{lemma}\label{Ubound}
Every solution of \eqref{LV} starting in $\Omega$ satisfies
\begin{equation}\label{Uinfty}
    \limsup_{t\to\infty} U(t)\le \frac{W^\infty}{1+W^\infty}\le \frac{r_1}{1+r_1}
\end{equation}
where $U=\sum_i (B_i+s_iP_i)$ and $W=\sum_i r_iB_i$ are defined by \eqref{U}.
\end{lemma}

\begin{proof}
Apply the fluctuation lemma, e.g. Prop. A.14 \cite{ST}, to \eqref{U} and use \eqref{r} and the invariance of $\Omega$ to conclude that $\sum_jB_j\le 1$.
\end{proof}

\begin{proposition}\label{limsup}
If $P_i(0)>0$, then $(\sum_{j\le i}B_j)_\infty\le s_i$.

If $B_1(0),P_n(0)>0$, then
\begin{equation}\label{wpersist}
\left (\sum_j (r_1s_j+1)P_j \right)^\infty\ge r_1-1-r_1s_n>0.
\end{equation}

\begin{enumerate}
\item [(a)] If $(\sum_{j\le i}B_j)^\infty< s_i$ then $P_i(t)\to 0$.\\
\item [(b)] If $i<j$, $P_i(0)>0$, and if $(B_{i+1}+B_{i+2}+\cdots +B_j)^\infty<s_j-s_i$, then $P_j(t)\to 0$.\\
\item [(c)] If $i<j$, $B_i(0)>0$, and if $(P_i+P_{i+1}+\cdots +P_{j-1})^\infty<\frac{r_i-r_j}{1+r_1}$, then $B_j(t)\to 0$.\\
\end{enumerate}
\end{proposition}

\begin{proof}
The equation for $P_i$ implies that
$$
\frac{d}{dt}\log P_i^{s_i}(t)=\sum_{j\le i}B_j-s_i.
$$
If $(\sum_{j\le i}B_j)_\infty\le s_i$ is false, then $P_i(t)\to\infty$, a contradiction to boundedness of solutions.
Assertion (a) is transparent.

\eqref{sss} implies that $r_n(1-s_n)-1=r_n-1-r_ns_n>0$ and, with \eqref{s},\eqref{r} together, imply that $r_1-1-r_1s_n>0$. We have
$$
\frac{B_1'}{B_1}=r_1-1-r_1\sum_jB_j-\sum_j(r_1s_j+1)P_j
$$
and
$$
\frac{s_nP_n'}{P_n}=\sum_j B_j-s_n.
$$
Multiplying the second expression by $r_1$ and adding to the first gives
\begin{eqnarray*}
  \frac{d}{dt}\log B_1P_n^{r_1s_n} &=& \frac{B_1'}{B_1}+\frac{r_1s_nP_n'}{P_n} \\
   &=& r_1-1-r_1s_n-\sum_j (r_1s_j+1)P_j
  \end{eqnarray*}
\eqref{wpersist} follows since the alternative is that $B_1P_n$ is unbounded, a contradiction.

Proof of (b): if $i<j$, $P_i(0)>0, P_j(0)>0$, and if $(B_{i+1}+\cdots+B_j)^\infty<s_j-s_i$, then
\begin{eqnarray*}
  \frac{d}{dt}\log \frac{P_i^{s_i}(t)}{P_j^{s_j}(t)} &=& \frac{s_iP_i'}{P_i}-\frac{s_j P_j'}{P_j} \\
  &=& -(B_{i+1}+\cdots + B_j)+s_j-s_i \\
   &\ge& \epsilon, \ t\ge T
\end{eqnarray*}
for some $\epsilon, T >0$. Therefore, $\frac{P_i^{s_i}(t)}{P_j^{s_j}(t)}\to \infty$, which implies that $P_j(t)\to 0$ since
$P_i(t)$ is bounded.

proof of (c): assume that $B_i(0)>0, B_j(0)>0$, and  $(P_i+P_{i+1}+\cdots+P_{j-1})^\infty<\frac{r_i-r_j}{1+r_1}$. Then, recalling that $U=\sum_k(B_k+s_kP_k)$, we have
\begin{eqnarray*}
  \frac{d}{dt}\log \frac{B_i(t)}{B_j(t)} &=& \frac{B_i'}{B_i}-\frac{B_j'}{B_j} \\
  &=& (r_i-r_j)(1-U)-(P_i+P_{i+1}+\cdots +P_{j-1}) \\
   &\ge& (r_i-r_j)\left(\frac{1-\epsilon}{1+r_1}\right)-(P_i+P_{i+1}+\cdots +P_{j-1}), \ t\ge T
\end{eqnarray*}
where, by \eqref{Uinfty}, we can choose $\epsilon>0$ so small that $(P_i+P_{i+1}+\cdots +P_{j-1})^\infty<(r_i-r_j)(\frac{1-\epsilon}{1+r_1})$. It follows that
$B_i/B_j\to\infty$ which implies that $B_j(t)\to 0$.

\end{proof}

Note that case (b) and (c) of Proposition~\ref{limsup} extend Lemma~\ref{extinction}.

\begin{proposition}\label{weakpersist}
If $B_1(0)>0$, then $B_1^\infty\ge s_1$.

If $B_1(0)>0$ and $P_1(0)>0$, then
$$B_1^\infty\ge s_1, \ P_1^\infty\ge \min\{\frac{r_1-r_2}{1+r_1},  \frac{r_1-1-r_1s_1}{r_1s_1+1}\}.$$
\end{proposition}

\begin{proof} Assume the conclusion is false. Then $P_1\to 0$ by Proposition~\ref{limsup} (a).
If $P_i(0)=0$ for all $i$, then $B_1(t)\to 1-1/r_1\ge s_1$ by the classical chemostat theory, e.g. Theorem 3.2 in \cite{SW}, so we suppose that $P_i(0)>0$  for some $i$.
Let $k$ denote the smallest such integer $i$ for which $P_i(0)>0$.

If $k=1$, then, as noted above, $P_1\to 0$ and so $B_2\to 0$ by Proposition~\ref{limsup} (c). Then $P_2\to 0$
by Proposition~\ref{limsup} (a) or (b).

If $k=2$, then $P_1\equiv 0$ so $B_2\to 0$ by Lemma~\ref{extinction} since $B_1$ and $B_2$ share the same virus.
Since $(B_1+B_2)^\infty=B_1^\infty<s_1<s_2$, it follows that $P_2\to 0$ by Proposition~\ref{limsup} (a). Now we can
use Proposition~\ref{limsup} (c) to show $B_3\to 0$ and then Proposition~\ref{limsup} (a) or (b) to show $P_3\to 0$.

If $k>2$, then $P_1\equiv P_2\equiv \cdots \equiv P_{k-1}\equiv 0$ and $P_k(0)>0$. As $B_1,\cdots, B_{k-1}$ share the same virus, then $B_i\equiv 0$ or $B_i\to 0$ for $1<i\le k-1$ by  Lemma~\ref{extinction}.  $B_k\to 0$ by Proposition~\ref{limsup} (c). Then,
$P_k\to 0$ by Proposition~\ref{limsup} (a) since $(\sum_{j\le k}B_j)^\infty=B_1^\infty<s_1<s_k$.
So $B_{k+1}\to 0$ by Proposition~\ref{limsup} (c). Proposition~\ref{limsup} (a) or (b) implies that
$P_{k+1}\to 0$.

We see that for all values of $k$, $B_2,\cdots, B_{k+1}\to 0$ and $P_1,\cdots, P_{k+1}\to 0$. Successive additional applications of
Proposition~\ref{limsup} (a) or (b) and (c) then imply that $B_2,\cdots, B_{n}\to 0$ and $P_1,\cdots, P_n\to 0$. But, then
$$
B_1'/(r_1B_1)\ge 1-\frac{1}{r_1}-\epsilon -B_1>s_1+\epsilon-B_1,\ t\ge T
$$
for some $\epsilon>0$ and $T>0$ (recall that $s_1<1-1/r_1$). This implies that $B_1^\infty>s_1$, a contradiction.
This completes the proof of the first assertion.

Now, suppose that $B_1(0)>0, P_1(0)>0$ and $P_1^\infty<\frac{r_1-r_2}{1+r_1}$. Proposition~\ref{limsup} (c) implies that
$B_2\to 0$. By  Proposition~\ref{limsup} (b), $P_2\to 0$. Applying Proposition~\ref{limsup} (c) with $i=1$ and $j=3$,
as $(P_1+P_2)^\infty=P_1^\infty<\frac{r_1-r_2}{1+r_1}<\frac{r_1-r_3}{1+r_1}$, we conclude that $B_3\to 0$. Then,
Proposition~\ref{limsup} (b) implies that $P_3\to 0$. Clearly, we can continue sequential application of Proposition~\ref{limsup} (b) and (c)
to conclude that $B_i,P_i\to 0$ for $i>1$. Now, we may argue as in the proof of \eqref{wpersist}
\begin{eqnarray*}
  \frac{d}{dt}\log B_1P_1^{r_1s_1} &=& \frac{B_1'}{B_1}+\frac{r_1s_1P_1'}{P_1} \\
   &=& r_1-1-r_1s_1-(r_1s_1+1)P_1-\hbox{terms that go to zero}
  \end{eqnarray*}
to conclude that $P_1^\infty\ge \frac{r_1-1-r_1s_1}{1+s_1r_1}$.

\end{proof}

The following is a slight modification of Theorem 5.2.3 in \cite{HS}.

\begin{lemma}\label{Hofbauer}
Let $x(t)$ be a bounded positive solution of the Lotka-Volterra system
$$
x_i'=x_i(r_i+\sum_{j=1}^n a_{ij}x_j ),\ 1\le i\le n
$$
and suppose there exists $k<n$ and $m,M,\delta>0$ such that $m\le x_i(t)\le M,\ 1\le i\le k,\ t>0$, $x_{k+1}(t)\le \delta,\ t>0$,
and $x_j(t)\to 0$ for $j>k+1$. Suppose also that the $k\times k$ subsystem obtained by setting $x_j=0, j>k$ has a unique positive
equilibrium $p=(p_1,p_2,\cdots,p_k)$. Then
$$
\liminf_{T\to\infty}\frac{1}{T}\int_0^T x_i(t)dt=p_i+O(\delta),\ 1\le i\le k.
$$
The same expression holds for the limit superior.
\end{lemma}

\begin{proof}
As in \cite{HS} Thm 5.2.3, we have that $z_j(T)=\frac{1}{T}\int_0^T x_j(t)dt$ satisfies
$$
\frac{\log x_i(T)-\log x_i(0)}{T}=\sum_{i=1}^k a_{ij}(z_j(T)-p_j)+a_{i(k+1)}z_{k+1}(T)+\sum_{j>k+1}a_{ij}z_j(T).
$$
for $i=1,2,\cdots,k$. As $T\to\infty$, the left hand side converges to zero and so does the final sum on the right since $x_j\to 0$ for $j>k+1$. The $k\times k$ matrix $\tilde A=(a_{ij})_{1\le i,j\le k}$ is invertible by hypothesis so we may write the above in vector form
as
$$
z(T)=p+(\tilde A)^{-1}(O(1/T)-z_{k+1}(T)A_{k+1})
$$
where $O(1/T)\to 0$ and $A_{k+1}$ is the first $k$ entries of the $k+1$ column of $A=(a_{ij})_{1\le i,j\le n}$. As $0\le z_{k+1}\le \delta$,
it follows that $\|z(T)-p\|\le c \delta$ for some $c$ and all large $T$. The result follows.
\end{proof}

Our main result follows.

\begin{theorem}\label{persist}
Let $1\le k\le n$.
\begin{enumerate}
  \item [(a)] There exists $\epsilon_k>0$ such that if $B_i(0)>0,\ 1\le i\le k$ and $P_j(0)>0,\ 1\le j\le k-1$, then
  $$
  B_{i,\infty}\ge \epsilon_k,\ 1\le i\le k \ \hbox{and}\ P_{j,\infty}\ge \epsilon_k,\ 1\le j\le k-1.
  $$
  \item [(b)] There exists $\epsilon_k>0$ such that if $B_i(0)>0, P_i(0)>0,\ 1\le i\le k$, then
  $$
  B_{i,\infty}\ge \epsilon_k,\ P_{i,\infty}\ge \epsilon_k,\ 1\le i\le k.
  $$
\end{enumerate}
\end{theorem}

\begin{proof} We use the notation $[B_i]_t\equiv \frac{1}{t}\int_0^t B_i(s)ds$.
Our proof is by mathematical induction using the ordering of the $2n$ cases as follows
$$
(a,1)<(b,1)<(a,2)<(b,2)<\cdots<(a,n)<(b,n)
$$
where $(a,k)$ denotes case (a) with index $k$.

The cases $(a,1)$ and $(b,1)$ follow immediately from Proposition~\ref{weakpersist} and by the general result that weak uniform
persistence implies strong uniform persistence under suitable compactness assumptions. See Prop. 1.2 in \cite{T} or Corollary 4.8 in \cite{ST} with persistence function
$\rho=\min\{B_1,P_1\}$ in case $(b,1)$. Note that our state space is compact.

For the induction step, assuming that $(a,k)$ holds, we prove that $(b,k)$ holds and assuming that $(b,k)$ holds, we prove that $(a,k+1)$ holds.

We begin by assuming that $(a,k)$ holds and prove that $(b,k)$ holds. We consider solutions satisfying $B_i(0)>0, P_i(0)>0$ for $1\le i\le k$.
Note that other components $B_j(0)$ or $P_j(0)$ for $j>k$ may be positive or zero, we make not assumptions.
As $(a,k)$ holds, there exists $\epsilon_k>0$ such that $B_{i,\infty}\ge \epsilon_k,\ 1\le i\le k$ and $P_{i,\infty}\ge \epsilon_k,\ 1\le i\le k-1$. We need only show the
existence of $\delta>0$ such that $P_{k,\infty}\ge \delta$ for every solution with initial values as described above. In fact, by the above-mentioned result
that weak uniform persistence implies strong uniform persistence, it suffices to show that $P_k^\infty \ge \delta$.

If $P_k^\infty<\frac{r_k-r_{k+1}}{1+r_1}$, then $B_{k+1}\to 0$ by Proposition~\ref{limsup} (c). Then, by Proposition~\ref{limsup} (b), $P_{k+1}\to 0$.
Clearly,  we may sequentially apply Proposition~\ref{limsup} (b) and (c) to show that $B_j\to 0, P_j\to 0$ for $j\ge k+1$.

If there is no $\delta>0$ such that $P_k^\infty \ge \delta$ for every solution with initial data as described above, then for every
$\delta>0$, we may find a solution with such initial data such that $P_k^\infty < \delta$. By a translation of time, we may assume that
$P_k(t)\le \delta,\ t\ge 0$ for $0<\delta<\frac{r_k-r_{k+1}}{1+r_1}$ to be determined later. Then $B_j,P_j\to 0, j\ge k+1$.
Now, as $(a,k)$ holds, we may apply Lemma~\ref{Hofbauer}. The subsystem with $B_i=0, \ k+1\le
i\le n$ and $P_i=0,\ k\le i\le n$ has a unique positive equilibrium
by Proposition~\ref{positivequil}. See Remark~\ref{Qn}. The equation
$$
\frac{s_kP_k'}{P_k}=\sum_{j\le k}B_j-s_k
$$
implies that
$$
\frac{1}{t}\log \frac{P_k^{s_k}(t)}{P_k^{s_k}(0)}=\sum_{j\le k}[B_j]_t -s_k.
$$
By \eqref{equiln} and Lemma~\ref{LV}, we have for large $t$
$$
\sum_{j\le k}[B_j]_t=\sum_{j\le k}B_j^\dag+O(\delta)=s_k+q+O(\delta)
$$
where $q=(1-\frac{1+Q_k}{r_k})>0$. On choosing $\delta$ small enough and an appropriate solution, then $\sum_{j\le k}[B_j]_t -s_k>q/2$
for large $t$, implying that $P_k\to +\infty$, a contradiction. We have proved that $(a,k)$ implies $(b,k)$.

Now, we assume that $(b,k)$ holds and prove that $(a,k+1)$ holds. We consider solutions satisfying $B_i(0)>0, P_i(0)>0$ for $1\le i\le k$ and $B_{k+1}(0)>0$.
As $(b,k)$ holds by assumption, and following the same arguments as in the previous case, we only need to show that there exists $\delta>0$ such that $B_{k+1}^\infty\ge \delta$ for all solutions with initial data
as just described.

If $B_{k+1}^\infty<s_{k+1}-s_k$, then $P_{k+1}\to 0$ by Proposition~\ref{limsup} (b) and then $B_{k+2}\to 0$ by Proposition~\ref{limsup} (c).
This reasoning may be iterated to yield $B_i\to 0,\ k+2\le i\le n$ and $P_i\to 0,\ k+1\le i\le n$.

If there is no $\delta>0$ such that $B_{k+1}^\infty \ge \delta$ for every solution with initial data as described above, then for every
$\delta>0$, we may find a solution with such initial data such that $B_{k+1}^\infty < \delta$. By a translation of time, we may assume that
$B_{k+1}(t)\le \delta,\ t\ge 0$ for $0<\delta<s_{k+1}-s_k$ to be determined later. Then $B_j,P_j\to 0, j\ge k+2$ and $P_{k+1}\to 0$.
Now, using that $(b,k)$ holds, we apply Lemma~\ref{Hofbauer}. The subsystem with $B_i=0, P_i=0 \ k+1\le
i\le n$ has a unique positive equilibrium
by Proposition~\ref{positivequil}. See Remark~\ref{Qn}.
The equation for $B_{k+1}$ is
$$
\frac{B_{k+1}'}{B_{k+1}}=r_{k+1}-1-r_{k+1}\sum_{j=1}^k(B_j+s_jP_j)-r_{k+1}\sum_{j=k+1}^n(B_j+s_jP_j)-\sum_{j=k+1}^nP_j
$$
Integrating, we have
$$
\frac{1}{t}\log \frac{B_{k+1}(t)}{B_{k+1}(0)}=r_{k+1}-1-r_{k+1}\sum_{j=1}^k([B_j]_t+s_j[P_j]_t)-r_{k+1}[B_{k+1}]_t+O(1/t)
$$
By \eqref{equil} and Lemma~\ref{Hofbauer}, we have that for all large $t$
$$
\sum_{j=1}^k([B_j]_t+s_j[P_j]_t)= \sum_{j=1}^k(B_j^*+s_j P_j^*)+O(\delta)=\frac{Q_k}{1+Q_k}+O(\delta).
$$
Since $B_{k+1}(t)\le\delta$, $[B_{k+1}]_t=O(\delta)$. Hence, for large $t$
$$
\frac{1}{t}\log \frac{B_{k+1}(t)}{B_{k+1}(0)}=\frac{r_{k+1}}{1+Q_k}-1+O(\delta)+O(1/t).
$$
Now, $\frac{r_{k+1}}{1+Q_k}>\frac{r_{k+1}}{1+Q_{k+1}}>1$ so by choosing $\delta$ sufficiently small and an appropriate solution, we can ensure that
the right hand side is bounded below by a positive constant for all large $t$, implying that $B_{k+1}(t)$ is unbounded.
This contradiction completes our proof that $(b,k)$ implies $(a,k+1)$. Thus, our proof is complete by mathematical induction.
\end{proof}

\begin{corollary}\label{meanvalue}
For every solution of \eqref{LV} starting with all components positive, we have that
\begin{equation}
\frac{1}{t}\int_0^t B_i(s)ds \to B_i^*, \ \frac{1}{t}\int_0^t P_i(s)ds\to P_i^*
\end{equation}
where $B_i^*, P_i^*$ are as in \eqref{equil}.

For every solution of \eqref{LV} starting with all components positive except $P_n(0)=0$, we have that
\begin{equation}
\frac{1}{t}\int_0^t B_i(s)ds \to B_i^\dag, \ \frac{1}{t}\int_0^t P_i(s)ds\to P_i^\dag
\end{equation}
where $B_i^\dag, P_i^\dag$ are as in \eqref{equiln}.
\end{corollary}

\begin{proof}
This follows from the previous theorem together with Theorem 5.2.3 in \cite{HS}.
\end{proof}

\end{document}